\documentclass[%
 reprint,
 amsmath,amssymb,
 aps,
onecolumn, nofootinbib]{revtex4-2}

\usepackage{amsthm}
\usepackage{amsmath}
\usepackage{hyperref}
\usepackage{graphicx}
\usepackage{dcolumn}
\usepackage{bm}
\usepackage{xcolor}
\usepackage{nicematrix}
\usepackage{enumitem}
\usepackage{varwidth}
\usepackage{qcircuit}
\usepackage{afterpage}
\usepackage{fontawesome}
\usepackage{float}
\makeatletter
\let\newfloat\newfloat@ltx
\makeatother
\usepackage{algorithm}
\usepackage{algpseudocode}
\newcommand{\boxR}[2]{$\Qcircuit @C=1em @R=0em {
 \gate{\!R_P\!} &
}\!\!\!\!\!(H)
$}

\newcommand{\bra}[1]{\langle #1 |}
\newcommand{\ket}[1]{| #1 \rangle}

\newtheorem{theorem}{Theorem}
\newtheorem{corollary}[theorem]{Corollary}
\newtheorem{lemma}[theorem]{Lemma}

\begin{document}

\preprint{APS/123-QED}

\title{Generalized Quantum Signal Processing}

\author{Danial Motlagh}
\affiliation{
University of Toronto, Department of Computer Science, Toronto On, Canada
}

\author{Nathan Wiebe}
\affiliation{
University of Toronto, Department of Computer Science, Toronto On, Canada
}%
\affiliation{
Pacific Northwest National Laboratory, Richland WA, USA
}%
\affiliation{
Canadian Institute for Advanced Research, Toronto On, Canada
}%

\begin{abstract}
Quantum Signal Processing (QSP) and Quantum Singular Value Transformation (QSVT) currently stand as the most efficient techniques for implementing functions of block encoded matrices, a central task that lies at the heart of most prominent quantum algorithms. However, current QSP approaches face several challenges, such as the restrictions imposed on the family of achievable polynomials and the difficulty of calculating the required phase angles for specific transformations. In this paper, we present a Generalized Quantum Signal Processing (GQSP) approach, employing general SU(2) rotations as our signal processing operators, rather than relying solely on rotations in a single basis. Our approach lifts all practical restrictions on the family of achievable transformations, with the sole remaining condition being that $|P|\leq 1$, a restriction necessary due to the unitary nature of quantum computation. Furthermore, GQSP provides a straightforward recursive formula for determining the rotation angles needed to construct the polynomials in cases where $P$ and $Q$ are known. In cases where only $P$ is known, we provide an efficient optimization algorithm capable of identifying in under a minute of GPU time, a corresponding $Q$ for polynomials of degree on the order of $10^7$. We further illustrate GQSP simplifies QSP-based strategies for Hamiltonian simulation, offer an optimal solution to the $\epsilon$-approximate fractional query problem that requires $\mathcal{O}\left(\frac{1}{\delta} + \log(\large\frac{1}{\epsilon})\right)$ queries to perform where $\mathcal{O}(1/\delta)$ is a proved lower bound, and introduces novel approaches for implementing bosonic operators. Moreover, we propose a novel framework for the implementation of normal matrices, demonstrating its applicability through synthesis of diagonal matrices, as well as the development of a new algorithm for convolution through synthesis of circulant matrices using only $\mathcal{O}(d \log{N} + \log^2N)$ 1 and 2-qubit gates for a filter of lengths $d$.
\end{abstract}

\maketitle


\section{Introduction}

In recent years, significant advancements in quantum algorithm theory have revealed a powerful, overarching insight: the majority of prominent quantum algorithms, such as Hamiltonian simulation \cite{berry2015hamiltonian,low2017hamiltonianQSP,low2019hamiltonian, LCU}, quantum search \cite{Yoder2014FixedpointSearch, Grover2005FixedpointSearch, Magniez_2011}, factoring \cite{factoring}, and quantum walks \cite{Szegedy2004QMC, Magniez_2011}, can be fundamentally reduced to the central task of implementing a matrix function of a Hamiltonian, $f(H)$, this was shown in \cite{Martyn2021GrandUnification, QSVT}. Over the years, numerous techniques have been developed for constructing such functions, including phase estimation methods such as the HHL algorithm \cite{HLL}, linear combination of unitaries (LCU) \cite{LCU}, and quantum signal processing (QSP) \cite{low2017hamiltonianQSP, dong2022ground,dong2021efficient, rossi2022multivariable, rossi2023quantum, wang2022quantum}. Among these, QSP stands out as the most versatile approach to date, capable of approximating a wide range of matrix functions through eigenvalue or singular value transformations of $H$, while requiring a minimal number of ancilla qubits. The basic idea of QSP is to build a polynomial approximation of the desired function by assuming oracular access to a unitary $U$ which encodes $H$. QSP has been demonstrated to yield asymptotically optimal Hamiltonian simulation algorithms \cite{low2017hamiltonianQSP}. QSP's applicability was then further extended to the case of non-square matrices by the QSVT technique \cite{QSVT}. Together, QSP and QSVT have provided an abstract formalism that facilitates the efficient implementation of a wide range of linear algebraic operations and transformations, creating a new language with greater expressivity for the construction of quantum algorithms.

Despite its success, QSP still suffers from a number of severe limitations. Most notable of them being the restrictions it places on the family of polynomials that can be built using it. For instance, in order to build an arbitrary polynomial using the original QSP one has to split the polynomial into four parts by first breaking the polynomial into its real and imaginary parts, and then further breaking down each of those into their respective even and odd components. One would then combine all parts together using linear combination of unitaries (LCU) \cite{LCU}, and furthermore perform a variant of amplitude amplification \cite{Berry2014OAA} on the resulting circuit.  This requires that we triple the number of operations needed to perform Hamiltonian simulation relative to what we would need if we could implement Hamiltonian dynamics directly.  Lastly, while the original QSP guarantees the existence of a set of parameters leading to polynomials satisfying the specified restrictions, finding those parameters has proved to be a complicated (if computationally efficient) task~\cite{haah2019product,dong2021efficient}. In this paper, we present an algorithm, which we term Generalized Quantum Signal Processing (GQSP), that overcomes all the above limitations leading to a more general framework for the central task of implementing functions of Hamiltonians.

Our GQSP algorithm provides us with the ability to build polynomials of unitary matrices using a single ancilla qubit with the only restriction being that its norm is not greater than 1 on the complex unit circle as demonstrated in Corollary~\ref{cor:PossiblePs}. Since GQSP does not restrict us to a fixed parity for the polynomial, it gives us the ability to approximate functions of a Hamiltonian $H$ without the need for LCU that appears in some applications of QSP~\cite{QSVT}.  Furthermore, in Section~\ref{sec:CalculatingPhases} we demonstrate that computing the rotation angles required to build a given polynomial is much more efficient and conceptually simpler within this framework. This extends the scope of QSP along two different axes. We then utilize this technique in Section~\ref{sec:PhaseFunctions} to develop a conceptually simplified formulation of qubitization-based Hamiltonian simulation, and give a provably optimal algorithm for the fractional query problem as special cases of what we call Phase Functions. In Section~\ref{sec:FourierBasisForMatrices} we introduce a powerful idea that extends the concept of Fourier decomposition to normal matrices using polynomials of a special unitary matrix. More precisely, we will demonstrate that any normal matrix can be written as a polynomial of the root of unity unitary in its eigenbasis. We then explore the utility of this idea by giving a general framework for synthesizing diagonal and convolution matrices. But first, we give a review of the original quantum signal processing algorithm in Section~\ref{sec:RevQSP} which sets the stage for our method introduced in Section ~\ref{sec:GQSP}.

\section{Review of Quantum Signal Processing}\label{sec:RevQSP}
The quantum signal processing algorithm constructs a function of $\mathcal{H}$ by interleaving applications of a signal operator with signal processing operators. The signal operator encodes $\mathcal{H}$ using controlled applications of $U = e^{i\mathcal{H}}$, where an ancillary qubit acts as the control. The signal processing operations, on the other hand, consist of single-qubit operations performed on the ancillary qubit. There are two conventions that are commonly used for the formulation of quantum signal processing. One where the signal operator is expressed in the $\sigma_x$ basis and the signal processing is done in the $\sigma_z$ basis, and another where the signal operator is expressed in the $\sigma_z$ basis and the signal processing is done in the $\sigma_x$ basis. The signal operator in the $\sigma_z$ basis is defined as:
\begin{equation}
    A_{\sigma_z} := 
    \begin{bNiceMatrix}[columns-width=auto]
        e^{i\mathcal{H}} & 0\\
        0 & e^{-i\mathcal{H}}\\
    \end{bNiceMatrix}
\end{equation}
We can express the signal operator in the $\sigma_x$ basis by applying a Hadamard gate to the both sides of the ancilla qubit:
\begin{equation}\label{eq:Original_QSP_Blockencoding}
    A_{\sigma_x} :=  \begin{bNiceMatrix}[columns-width=auto]
    \cos(\mathcal{H}) & i\sin(\mathcal{H})\\
    i\sin(\mathcal{H}) & \cos(\mathcal{H})\\
    \end{bNiceMatrix}
    =(H \otimes I) A_{\sigma_z} (H \otimes I)\nonumber
\end{equation}

This shows that, at a high level, both of these formalisms are conceptually identical although there are distinctions between the forms of the polynomials that can be constructed in both bases.
The results of QSP are often stated in the $A_{\sigma_x}$ formalism as:
\begin{theorem}[Quantum Signal Processing in $\sigma_x$ Basis]\label{thm:QSP_x}
$\forall d\in \mathbb{N}$ and $\Vec{\phi} \in \mathbb{R}^{d + 1}$:
\begin{equation}
    \left( \prod\limits_{j=1}^{d} R_z(\phi_{j}) \, A_{\sigma_x} \right) R_z(\phi_{0}) = 
    \begin{bNiceMatrix}[columns-width=auto]  
    P(\cos(\mathcal{H})) & iQ(\cos(\mathcal{H}))\sin(\mathcal{H})\\
    iQ^*(\cos(\mathcal{H}))\sin(\mathcal{H}) & P^*(\cos(\mathcal{H}))\\
    \end{bNiceMatrix}
\end{equation}
\[
\mbox{\Large$\Longleftrightarrow$}
\]
\begin{center}
    \begin{varwidth}{\textwidth}
        \begin{enumerate}[label=(\Roman*)]
            \item  $P$ and $Q$ are polynomials s.t. $\text{deg}(P) \leq d$ and $\text{deg}(Q)\leq d-1$.
            \item  P has parity d mod 2 and Q has parity (d - 1) mod 2.
            \item  $\forall x \in [-1, 1], \>\> |P(x)|^2 + (1-x^2)|Q(x)|^2 = 1$
        \end{enumerate}
    \end{varwidth}
\end{center}
\end{theorem}
$ \,$

In practice, one often only cares about building a particular polynomial $P$, so the question becomes, for what $P$s there exists a $Q$ such that they satisfy the conditions of Theorem~\ref{thm:QSP_x}? It turns out the restrictions above can be quite limiting. For example, the 3rd condition requires $|P(\pm 1)| = 1$. This limitation can be overcome by applying Hadamards on the ancilla before and after the QSP sequence, effectively converting into the $A_{\sigma_z}$ formalism. In this case, it can be shown that we can implement any real polynomial with parity $d$ mod 2 such that $\text{deg}(P) \leq d$, and
\begin{theorem}[Quantum Signal Processing in $\sigma_z$ Basis]\label{thm:QSP_z}

 $\forall P\in \mathbb{R}[x]$, with $\textrm{deg}(P) = d, \> \exists\Vec{\phi} \in \mathbb{R}^{d + 1}$ s.t:\\
\begin{equation}
\forall x\in \mathbb{T}, \> |P(x)|^2\leq 1\> \wedge\> \text{Parity}(P) = d \>\text{mod}\>2
\end{equation}
\[
\mbox{\Large$\Longleftrightarrow$}
\]
\begin{equation}
    \begin{bmatrix}
        P(e^{i\mathcal{H}}) & .\>\>\>  \\
        
        .\>\>\> & .\>\>\> \\
    \end{bmatrix} = \left( \prod\limits_{j=1}^{d} R_x(\phi_{j}) \, A_{\sigma_z} \right) R_x(\phi_{0})  
\end{equation}
\end{theorem}
\noindent Here $\mathbb{T}= \{x\in\mathbb{C}: |x| = 1\}$ is again the complex unit circle.

The above formalism indeed mitigates some previously discussed limitations. However, the requirement that the polynomial $P$ must be real and possess definite parity necessitates the combination of four separate instances of QSP using Linear Combination of Unitaries (LCU) to construct an arbitrary complex polynomial with indefinite parity. In the following section, we address these constraints with a novel approach. We propose a more generalized formulation of Quantum Signal Processing by allowing our signal processing operators to be arbitrary SU(2) rotations, effectively lifting the "realness" condition on the polynomials. Furthermore, we eliminate the parity restriction by simplifying our signal operator. Combined together, our reformulation successfully eliminates all practical restrictions on $P$, which we will elaborate on and demonstrate in the next section.
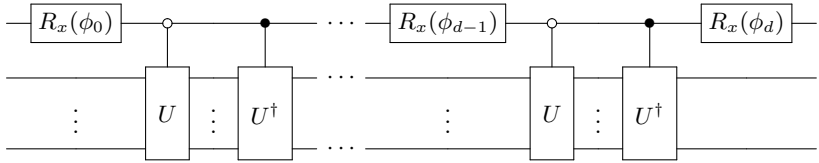
\begin{figure}[t!]
\[\Qcircuit @C=1em @R=1em {
& \gate{R_x(\phi_0)} & \ctrlo{1} & \qw & \ctrl{1} & \qw & \cdots & & \gate{R_x(\phi_{d-1})} & \ctrlo{1} & \qw &\ctrl{1} &\gate{R_x(\phi_d)}&\qw \\
& \qw & \multigate{2}{U} & \qw & \multigate{2}{U^\dagger} & \qw & \cdots & & \qw & \multigate{2}{U} & \qw & \multigate{2}{U^\dagger}&\qw&\qw \\
& \vdots & & \vdots & & & & & \vdots & & \vdots & \\
& \qw & \ghost{U} & \qw & \ghost{U^\dagger} & \qw & \cdots & & \qw & \ghost{U} & \qw & \ghost{U^\dagger}&\qw&\qw \\
}
\]
\caption{Quantum circuit for the original QSP approach where the angles $\phi$ are chosen to enact a polynomial transformation of $U$\label{fig:QSP_circuit}.}
\end{figure}

\section{Generalized Quantum Signal Processing}\label{sec:GQSP}

Our generalized quantum signal processing approach seeks to remove the limitations inherent in traditional QSP, thereby providing a more powerful mechanism for constructing functions of Hamiltonians.  Our approach generalizes QSP by using SU(2) rotations rather than $X$-rotations on the control ancilla similar to~\cite{2022arXiv220602826D, yu2022power}; however, our approach eschews the need to use $U^\dag$ and avoids the use of Laurent polynomials in the analysis.  Specifically, we use the convention that  the signal operator is a 0-controlled application of $U = e^{iH}$:
\begin{equation}
A = (\vert 0\rangle\langle0\vert \otimes U)+(\vert 1\rangle\langle1\vert \otimes I) = \begin{bmatrix}

U & 0  \\

0 & I \\
\end{bmatrix}\end{equation}
Next, instead of performing rotations in a single basis, we allow for signal processing operations to be arbitrary SU(2) rotations of the ancillary qubit:
\begin{equation}
R(\theta, \phi, \lambda) = \begin{bmatrix} e^{i(\lambda+\phi)}\cos(\theta) & e^{i\phi}\sin(\theta) \\
e^{i\lambda}\sin(\theta) & -\cos(\theta)  \\ \end{bmatrix} \otimes I\label{eq:R}
\end{equation}
Then by interleaving the above 2 operations as demonstrated in Figure ~\ref{fig:GQSP_circuit}, we can block encode polynomial transformations of the unitary matrix $U$ without further assumptions. Note that below we have no assumptions made on the parity of the polynomial transformations unlike those made by the Theorem~\ref{thm:QSP_x}. Lastly, notice that we only need to specify a single $\lambda$ since we can absorb the other instances into the instance of $\phi$ before them.\newpage
\begin{theorem}[Generalized Quantum Signal Processing]\label{thm:PolynomialOfUnitaries}
$\forall d\in \mathbb{N},\>\exists\> \vec{\theta}, \vec{\phi} \in \mathbb{R}^{d+1},\>\lambda \in \mathbb{R}$ s.t:\\

\begin{equation}
    \left( \prod_{j=1}^{d} \mbox{$\large {R(\theta_{j}, \phi_{j}, 0) A}$} \right) R(\theta_0, \phi_0, \lambda) = 
    \begin{bmatrix}
        P(U) & .\>\>\>  \\
        
        Q(U) & .\>\>\> \\
    \end{bmatrix}\\
\end{equation}
\[
\mbox{\Large$\Longleftrightarrow$}
\]
\begin{center}
    \begin{varwidth}{\textwidth}
        \begin{enumerate}
            \item  $P, \,Q\in \mathbb{C}[x]$ and $\text{deg}(P), \text{deg}(Q) \leq d$.
            \item  $\forall x\in \mathbb{T},\> |P(x)|^2 + |Q(x)|^2 = 1$.
        \end{enumerate}
    \end{varwidth}
\end{center}
\end{theorem}

\noindent Here $\mathbb{T}= \{x\in\mathbb{C}: |x| = 1\}$ is again the complex unit circle.

\par \vspace{\baselineskip}

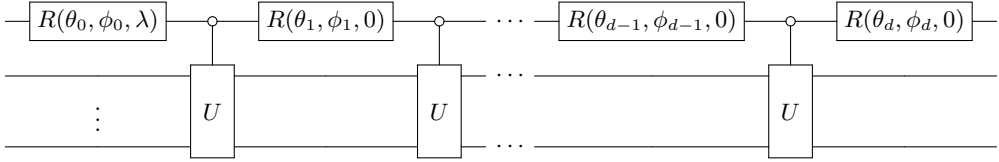
\begin{figure}[t!]
\[\Qcircuit @C=1em @R=1em {
& \gate{R(\theta_0, \phi_0, \lambda)} & \ctrlo{1} & \gate{R(\theta_1, \phi_1, 0)} & \ctrlo{1} & \qw & \cdots & & \gate{R(\theta_{d-1}, \phi_{d-1}, 0)} & \ctrlo{1} & \gate{R(\theta_d, \phi_d, 0)} & \qw \\
& \qw & \multigate{2}{U} & \qw & \multigate{2}{U} & \qw & \cdots & & \qw & \multigate{2}{U} & \qw & \qw \\
& \vdots & & & & & & & & & & & & \\
& \qw & \ghost{U} & \qw & \ghost{U} & \qw & \cdots & & \qw & \ghost{U} & \qw & \qw \\
}
\]
\caption{Quantum circuit for GQSP wherein the angles $\theta$ and $\phi$ are chosen to enact a polynomial transformation of $U$\label{fig:GQSP_circuit}.  Here $R(\theta,\phi,\lambda)$ is a general single qubit rotation as specified in~\eqref{eq:R}.}
\end{figure}
\begin{proof}
 Both directions are proven by induction on $d$.
To show the forward direction, we need to show inductively that for all $\{\theta_j\}$ and $\{\phi_j\}$ and $d$ both of the claimed restrictions are met.  The base case of $d=0$ is trivial since: 

\begin{equation}
    R(\theta_0, \phi_0, \lambda) = \begin{bmatrix} e^{i(\lambda+\phi)}\cos(\theta) I & e^{i\phi}\sin(\theta)I \\
    e^{i\lambda}\sin(\theta)I & -\cos(\theta) I \\ \end{bmatrix}
\end{equation}
 Next let $d \in \mathbb{N}_{>0}$ and assume the induction hypothesis holds for $d - 1$ producing $\hat{P}(U)$ and $\hat{Q}(U)$. Then:\\
\begin{equation}
\begin{bmatrix}
P(U) & .\>\>\>  \\

Q(U) & .\>\>\> \\
\end{bmatrix}  = 
\begin{bmatrix} e^{i\phi}\cos(\theta) U & e^{i\phi}\sin(\theta)I \\
\sin(\theta)U & -\cos(\theta) I \\ \end{bmatrix}
\begin{bmatrix}
\hat{P}(U) & .\>\>\>  \\

\hat{Q}(U) & .\>\>\> \\
\end{bmatrix}
\end{equation}
Thus:
\begin{equation}
P(U) = e^{i\phi}(\cos(\theta) U\hat{P}(U) + \sin(\theta)\hat{Q}(U))\\
\end{equation}
\begin{equation}
Q(U) = \sin(\theta) U\hat{P}(U) -\cos(\theta)\hat{Q}(U)\\
\end{equation}
 Since $\hat{P}$ and $\hat{Q}$ have degree $\leq d -1$, then $P$ and $Q$ must have degree $\leq d$. Furthermore, given the fact that all operations are unitary, property (2) trivially holds. This completes the proof for the forward direction.\\

The proof for the reverse direction involves showing by induction that any pair of polynomials satisfying the conditions of Theorem~\ref{thm:PolynomialOfUnitaries} can be constructed through an appropriate choice of rotation angles for the single qubit unitaries in Fig.~\ref{fig:GQSP_circuit}. The base case of $d=0$ is again trivial since $P$ and $Q$ would both be constants whose norm squared adds to 1. Any such pair can be written as $P = e^{i(\lambda+\phi)}\cos(\theta), \>Q = e^{i\phi}\sin(\theta)$ and implemented by $R(\theta_0, \phi_0, \lambda)$.
Next let $d \in \mathbb{N}_{>0}$, assume $P$ and $Q$ satisfy conditions (1) and (2) from Theorem~\ref{thm:PolynomialOfUnitaries}, and the induction hypothesis holds for $d - 1$. Specifically, we assume from our induction hypothesis that for any $\hat{P}$ and $\hat{Q}$ of degree at most $d-1$ satisfying:
$\forall x\in \mathbb{T},\> |\hat{P}(x)|^2 + |\hat{Q}(x)|^2 = 1$, there exists $\vec{\theta}, \vec{\phi} \in \mathbb{R}^{d},\>\lambda \in \mathbb{R}$ such that:
\begin{equation}\begin{bmatrix}
\hat{P}(U) & .\>\>\>  \\

\hat{Q}(U) & .\>\>\> \\
\end{bmatrix} =
\left( \prod_{j=1}^{d-1} \mbox{$\large {R(\theta_{j}, \phi_{j}, 0) A}$} \right) R(\theta_0, \phi_0, \lambda)
\end{equation}
Thus, it suffices to find $\theta_{d}$ and $\phi_{d}$ such that $\hat{P}$ and $\hat{Q}$ are of degree at most $d-1$ satisfying $|\hat{P}|^2 + |\hat{Q}|^2 = 1$ on $\mathbb{T}$, where $\hat{P}$ and $\hat{Q}$ are given by:
\begin{equation}
A^{\dag}R(\theta_{d}, \phi_{d}, 0)^{\dag}
\begin{bmatrix}
P(U) & .\>\>\>  \\

Q(U) & .\>\>\> \\
\end{bmatrix} =
\begin{bmatrix}
\hat{P}(U) & .\>\>\>  \\

\hat{Q}(U) & .\>\>\> \\
\end{bmatrix}
\end{equation}
\begin{equation}
=\begin{bmatrix} e^{-i\phi_d}\cos(\theta_d) U^{\dag} & \sin(\theta_d)U^{\dag} \\
e^{-i\phi_d}\sin(\theta_d)I & -\cos(\theta_d) I \\ \end{bmatrix}
\begin{bmatrix}
P(U) & .\>\>\>  \\

Q(U) & .\>\>\> \\
\end{bmatrix}
\end{equation}
Then:
\begin{equation}\label{eq:P_Hat}
    \hat{P}(U) = e^{-i\phi_d}\cos(\theta_d) U^{\dag}P(U) + \sin(\theta_d)U^{\dag}Q(U)
\end{equation}
\begin{equation}\label{eq:Q_Hat}
    \hat{Q}(U) = e^{-i\phi_d}\sin(\theta_d)P(U) -\cos(\theta_d) Q(U)
\end{equation}
First notice that regardless of our choice of $\theta_{d}$ and $\phi_{d}$, we have:
\begin{equation}
    \|\hat{P}(e^{it})\|^2 + \|\hat{Q}(e^{it})\|^2 = \hat{P}(e^{it})\hat{P}^*(e^{it}) + \hat{Q}(e^{it})\hat{Q}^*(e^{it})
\end{equation}
\begin{equation}
    = \cos(\theta_d)^2\|P(e^{it})\|^2 + \sin(\theta_d)^2\|Q(e^{it})\|^2 +\sin(\theta_d)^2\|P(e^{it})\|^2+\cos(\theta_d)^2\|Q(e^{it})\|^2
\end{equation}
\begin{equation}
    = \|P(e^{it})\|^2 + \|Q(e^{it})\|^2 = 1
\end{equation}
Hence, the norm condition is again satisfied trivially by the unitary nature of our operations. Thus, it remains to show $\hat{P}$ and $\hat{Q}$ are of degree at most $d-1$.\\
Let $a_i$ and $b_i$ be the coefficients of the $i$-degree term in $P$ and $Q$ respectively. By our assumption, we have that the target polynomials $P$ and $Q$ satisfy the norm condition, or equivalently for any real-valued $t$:
\begin{equation}\label{eq:Satisfaction_of_property2}
    \|P(e^{it})\|^2 + \|Q(e^{it})\|^2 = P(e^{it})P^*(e^{it}) + Q(e^{it})Q^*(e^{it}) = 1
\end{equation}
\begin{equation}
    \left(\sum_{n = 0}^d a_n e^{int}\right)\left(\sum_{n = 0}^d a_n^* e^{-int}\right)+
    \left(\sum_{n = 0}^d b _n e^{int}\right)\left(\sum_{n = 0}^d b_n^* e^{-int}\right) = 1
\end{equation}
Next, we need to argue about the structure of the coefficients in this expansion.  To begin we note that the inner product between the two parts from $P$ satisfies:
\begin{equation}
\left(\sum_{n = 0}^d a_n e^{int}\right)\left(\sum_{m = 0}^d a_m^* e^{-imt}\right) = \|\vec{a}\|^2 + \sum_n \sum_{m\ne n} a_n a_m^* e^{i(n-m)t}.
\end{equation}
Repeating the same expansion for the $b$ coefficients leads to the observation that for all $t$:
\begin{equation}
1 = \|\vec{a}\|^2+ \|\vec{b}\|^2 + \sum_n \sum_{m\ne n} a_n a_m^* e^{i(n-m)t} +b_n b_m^* e^{i(n-m)t}.\label{eq:expandSimp}
\end{equation}
We then note that the $t$-dependent sum must generically be zero because of the fact that the exponentials form an orthogonal function basis as seen by the Fourier transform as noted by the fact that
\begin{equation}
    \frac{1}{2\pi}\int_{-\infty}^\infty e^{i(p-q)x} \mathrm{d}x = \delta(p-q).
\end{equation}
where $\delta$ is the Dirac-delta function.  Thus because the functions form an orthogonal function space that does not contain the constant function, we cannot form the constant function out of a linear combination of non-constant exponentials, and the only way that we can satisfy~\eqref{eq:expandSimp} for all real-valued $t$ is if:
\begin{equation}
    \sum_{n = 0}^d (\|a_n\|^2 + \|b_n\|^2) = 1
\end{equation}
and
\begin{equation}\label{eq:cancelling_Terms}
    \forall i\in \mathbb{Z}, i\neq 0, \qquad \sum_{n} a_na_{n-i}^* = -\sum_{n}b_nb_{n-i}^*
\end{equation}
Let $s, s'\geq 0$ be the smallest degrees present in $P$ and $Q$ respectively, then notice that~\eqref{eq:cancelling_Terms} implies $\textrm{deg}(P)-s = \textrm{deg}(Q)-s'$ since otherwise setting $i = \max(\textrm{deg}(P)-s, \textrm{deg}(Q)-s')$ would result in one side of~\eqref{eq:cancelling_Terms}  being 0 and the other being a non-zero value. It's easy to see when $\textrm{deg}(P) > \textrm{deg}(Q)$ or $\textrm{deg}(P) < \textrm{deg}(Q)$ setting $\theta_d$ equal to $0$ or $\frac{\pi}{2}$ respectively will result in the desired $\hat{P}$ and $\hat{Q}$. Therefore, without loss of generality, we assume $\textrm{deg}(P) = \textrm{deg}(Q) = d$.\, which also implies $s=s'$. Hence using ~\eqref{eq:cancelling_Terms} we get:
\begin{equation}\label{eq:Orthogonality_Relation}
    a_d a^*_s = -b_d b^*_s
\end{equation}
Next, we choose $\theta_d = \tan^{-1}(\frac{|b_d|}{|a_d|})$ and $\phi_d = \textrm{Arg}(\frac{a_d}{b_d})$. Then by ~\eqref{eq:Orthogonality_Relation}:
\begin{equation}
e^{i\phi_d} \mbox{\Large$\frac{\cos(\theta_d)}{\sin(\theta_d)} = \frac{a_d}{b_d} = -\frac{b^*_s}{a^*_s}$}
\end{equation} 
Then it is easy to see after substituting our choices of $\theta_d$ and $\phi_d$ into~\eqref{eq:P_Hat} and~\eqref{eq:Q_Hat} the $s$ degree terms of $P$ and $Q$ in $\hat{P}$ cancel (no negative degree terms are introduced) and $d$ degree terms of $P$ and $Q$ in $\hat{Q}$ also cancel. Thus, $\textrm{deg}(\hat{P})=\textrm{deg}(\hat{Q})=d-1$ as desired.
\end{proof}

 Theorem~\ref{thm:PolynomialOfUnitaries} specifies the necessary and sufficient conditions for constructible pairs of polynomials $P$ and $Q$ using our method. However, in practice, one often only cares about building a particular polynomial $P$, so the question becomes: ``For what $P$s there exists a $Q$ such that they satisfy the conditions of theorem~\ref{thm:PolynomialOfUnitaries}?''\\
 In the following theorem we show that, as long as $|P|^2 \leq 1$ on $\mathbb{T}$, there exists such a $Q$.
\begin{theorem}\label{thm:QExistence}

 $\forall P\in \mathbb{C}[x]$, we have:
\begin{equation}
\forall x\in \mathbb{T}, \>|P(x)|^2 \leq 1
\end{equation}
\[
\mbox{\Large$\Longleftrightarrow$}
\]
\begin{equation}
\exists Q\in \mathbb{C}[x]\>\> s.t. \>\>\textrm{deg}(P)=\textrm{deg}(Q) \>\> \wedge \>\> \forall x\in \mathbb{T},\> |P(x)|^2 + |Q(x)|^2 = 1
\end{equation}
\end{theorem}

Before proving Theorem~\ref{thm:QExistence}, notice that this leads to the following corollary, which states that we can build any arbitrary (appropriately scaled) polynomial $P$ of $U$ using our technique.
\begin{corollary}\label{cor:PossiblePs}
 $\forall P\in \mathbb{C}[x]$, with $\text{deg}(P) = d$ if:\\
\begin{equation}
\forall x\in \mathbb{T}, \>|P(x)|^2 \leq 1
\end{equation}
Then $\exists\> \vec{\theta}, \vec{\phi} \in \mathbb{R}^{d+1}, \> \exists \lambda \in \mathbb{R}$ such that:
\begin{equation} 
\begin{bmatrix}
P(U) & .\>\>\>  \\
.\>\>\> & .\>\>\> \\
\end{bmatrix} = \left( \prod_{j=1}^{d} \mbox{$\large {R(\theta_{j}, \phi_{j}, 0) A}$} \right) R(\theta_0, \phi_0, \lambda)
\end{equation}
\end{corollary}
Corollary~\ref{cor:PossiblePs} is one of the most significant results in this paper, as it clearly establishes the superior expressivity of our method over traditional QSP when contrasted with ~\ref{thm:QSP_z}.

\begin{proof}[Proof of Theorem~\ref{thm:QExistence}]$ $\newline

 $\Longrightarrow$:
Let $P\in \mathbb{C}[x]$ with $|P|^2 \leq 1$ on $\mathbb{T}$, and define:
\begin{equation}
T(\theta)=|P(e^{i\theta})|^2
\end{equation}
Then T is a non negative trigonometric polynomial with terms from $-\textrm{deg}(P)$ to $\textrm{deg}(P)$ that satisfies $T \le 1$. Define $H=1-T$ as another non negative trigonometric polynomial of degree $\textrm{deg}(P)$. It remains to show there exists $Q$ a polynomial of degree $\textrm{deg}(P)$ such that $H(\theta)=|Q(e^{i\theta})|^2$.\\

\noindent Let $d = \textrm{deg}(P)$, then we can re-write $H$ as:
\begin{equation}
H(\theta)=e^{-id\theta}R(e^{i\theta})
\end{equation}
For some polynomial $R$ of degree $2d$. Since $H =H^*$, we have for all $|z|=1$:
\begin{equation}
z^{2d}R^*(z)=z^{2d}R^*(\frac{1}{z^*})=R(z)
\end{equation}
Using the identity theorem for analytic functions, this equality must hold for all $z \in \mathbb C - \{0\}$, making $R$ a self-inversive polynomial. Therefore, roots of $R$ are either on the unit circle, or grouped in pairs $(w_i, \frac{1}{w_i^*})$ for $|w_i|>1$.\\

 \noindent Notice $(e^{i\theta}-w_i)(e^{i\theta}-\frac{1}{w_i^*})=-\frac{e^{i\theta}}{w_i^*}\|e^{i\theta}-w_i\|^2$, so by grouping the roots $(w_i, \frac{1}{w_i^*})$ together, we get:\\
\begin{equation}
R(e^{i\theta})=c(-1)^me^{im\theta}\prod_{k=1}^{m}\frac{1}{w_k^*}\>\prod_{k=1}^{m}\|e^{i\theta}-w_k\|^2\prod_{k=1}^{2d-2m}(e^{i\theta}-e^{i\theta_k})
\end{equation}
Where $m$ is the number of such pairs. We can then write:
\begin{equation}
H(\theta)=e^{-id\theta}R(e^{i\theta})=|G(e^{i\theta})|^2\hat{H}(\theta)
\end{equation}
where $G$ is a polynomial of degree $m$ and $\hat{H}$ is a non negative trigonometric polynomial with roots only on the unit circle, that maps the unit circle into $[0, \infty)$. Just like $H$, we can re-write $\hat{H}$ as:

\begin{equation}
\hat{H}(\theta)=e^{i(d-m)\theta}\hat{R}(e^{i\theta})
\end{equation}

 For some polynomial $\hat{R}$ of degree $2d - 2m$. Extending $\hat{H}$ to an analytic function on a small annulus including the unit circle, we can use the local form of the analytic function near the zeros. This implies a curve passing through a zero of $\hat{R}$ is sent either to a curve through 0, or a line segment ending at 0 around a small neighborhood of the zero depending on the parity of the zero. However, since $\hat{H}$ has all roots on the unit circle, then the unit circle itself passes through all zeros and is mapped into $[0, \infty)$, which implies all zeros have even multiplicity. Hence, we can re-write $\hat{R}$ as:

\begin{equation}
\hat{R} = \hat{G}^2
\end{equation}

For some polynomial $\hat{G}$ of degree $d - m$. Then on the complex unit circle, we have:

\begin{equation}
\hat{H}(\theta) = |\hat{H}(\theta)|= |e^{i(d-m)}\hat{R}(e^{i\theta})| = |\hat{G}(e^{i\theta})^2| = |\hat{G}(e^{i\theta})|^2
\end{equation}
Thus:
\begin{equation}
    H(\theta)=|G(e^{i\theta})|^2\cdot|\hat{G}(e^{i\theta})|^2= |(G\hat{G})(e^{i\theta})|^2
\end{equation}
Then letting $Q = G\hat{G}$ we have:
\begin{equation}
    H(\theta) = |Q(e^{i\theta})|^2
\end{equation}
\begin{equation}
\Longrightarrow\>\> |P(e^{i\theta})|^2 + |Q(e^{i\theta})|^2 = 1
\end{equation}
and:
\begin{equation}
\Longrightarrow\>\> \textrm{deg}(Q) = \textrm{deg}(G)+\textrm{deg}(\hat{G}) = m + d - m = d = \textrm{deg}(P)
\end{equation}
 which proves the forward direction of the proof.

 $\Longleftarrow$: This is trivial since $|Q(e^{i\theta})|^2\geq 0$.\\
\end{proof}

So far we have only discussed the construction of polynomials of positive degree, however, the framework presented here can also be used to implement polynomials with any ratio of positive and negative degrees by making use of a secondary signal operator which we'll define as: 
\begin{equation}
    A^{'} = (\vert 0\rangle\langle0\vert \otimes I)+(\vert 1\rangle\langle1\vert \otimes U^\dag) = \begin{bmatrix}

I & 0  \\

0 & U^\dag \\
\end{bmatrix}
\end{equation}
Particularly, if the set of rotation angles $\vec{\theta},\, \vec{\phi}$, and $\lambda$ lead to the polynomials $P(x) = \sum_{n = 0}^d a_n\, e^{inx}$ and $Q(x) = \sum_{n = 0}^d b_n\, e^{inx}$, using our original signal operator $A$, then we can construct polynomials $P^{'}(x) = e^{-ikx} P(x) = \sum_{n = -k}^{d-k} a_{n+k}\, e^{inx}$ and $Q^{'}(x) = e^{-ikx} Q(x) = \sum_{n = -k}^{d-k} b_{n+k}\, e^{inx}$ for $k\leq d$, using the same set of rotation angles by replacing any $k$ instances of $A$ with the complementary signal operator $A^{'}$.
\begin{theorem}[Polynomials With Negative Powers]\label{thm:positive_negative}
$\forall d,k\in \mathbb{N},\,\forall \vec{\theta}, \vec{\phi} \in \mathbb{R}^{d+1},\>\lambda \in \mathbb{R}$ and $k\leq d$ we have:\\
{\small
\begin{equation}
    \begin{bmatrix}
        P^{'}(U) & .\>\>\>  \\
        Q^{'}(U) & .\>\>\> \\
    \end{bmatrix}=
\left[\prod_{j=1}^{k} \mbox{$\large {R(\theta_{d-k+j}, \phi_{d-k+j}, 0) A^{'}}$}\right]
\left[\prod_{j=1}^{d-k} \mbox{$\large {R(\theta_{j}, \phi_{j}, 0) A}$} \right]R(\theta_0, \phi_0, \lambda)
\end{equation}
}
\noindent If and only if:
{\small
\begin{equation}
    \begin{bmatrix}
        P(U) & .\>\>\>  \\
        
        Q(U) & .\>\>\> \\
    \end{bmatrix}=\left[\prod_{j=1}^{d} \mbox{$\large {R(\theta_{j}, \phi_{j}, 0) A}$} \right] R(\theta_0, \phi_0, \lambda)
\end{equation}
}
For $P^{'}(e^{ix}) = e^{-ikx} P(e^{ix})$ and $Q^{'}(e^{ix}) = e^{-ikx} Q(e^{ix})$
\end{theorem}

\begin{proof}[Proof of Theorem~\ref{thm:positive_negative}]$ $\\

\noindent First notice:
\begin{equation}
    A^{'} = \left( I\otimes U^\dag\right)A
\end{equation}
Hence we have 
\begin{align}
    &\left[ \prod_{j=1}^{k} \mbox{$\large {R(\theta_{d-k+j}, \phi_{d-k+j}, 0) A^{'}}$} \right]\left[\prod_{j=1}^{d-k} \mbox{$\large {R(\theta_{j}, \phi_{j}, 0) A}$} \right] R(\theta_0, \phi_0, \lambda) \\&= \left[ \prod_{j=1}^{k} \mbox{$\large {R(\theta_{d-k+j}, \phi_{d-k+j}, 0) \left( I\otimes U^\dag\right)A}$} \right]\left[ \prod_{j=1}^{d-k} \mbox{$\large {R(\theta_{j}, \phi_{j}, 0) A}$} \right] R(\theta_0, \phi_0, \lambda)
\end{align}
Since $\left( I\otimes U^\dag\right)$ commutes with both $A$ and $R(\theta, \phi, 0)$ we then have:
\begin{equation}
    = \left( I\otimes U^\dag\right)^k\left[ \prod_{j=1}^{k} \mbox{$\large {R(\theta_{d-k+j}, \phi_{d-k+j}, 0) A}$} \right]\left[ \prod_{j=1}^{d-k} \mbox{$\large {R(\theta_{j}, \phi_{j}, 0) A}$} \right] R(\theta_0, \phi_0, \lambda)
\end{equation}
\begin{equation}
    = \left( I \otimes U^{-k}\right)\left[ \prod_{j=1}^{d} \mbox{$\large {R(\theta_{j}, \phi_{j}, 0) A}$} \right] R(\theta_0, \phi_0, \lambda)
\end{equation}
\begin{equation}
    = \begin{bmatrix}

U^{-k} & 0  \\

0 & U^{-k} \\
\end{bmatrix}\begin{bmatrix}
        P(U) & .\>\>\>  \\
        
        Q(U) & .\>\>\> \\
    \end{bmatrix}=\begin{bmatrix}
        P^{'}(U) & .\>\>\>  \\
        
        Q^{'}(U) & .\>\>\> \\
    \end{bmatrix}
\end{equation}
    
\end{proof}

In this section, we gave a detailed description of our generalized QSP technique and proved its first advantage over traditional QSP. Namely, we showed that our method overcomes the restrictions placed on the family of polynomials that can be built using QSP, which eliminates the need for LCU and oblivious amplitude amplification (OAA). In the next section, we will demonstrate the second advantage of our formalism over traditional QSP, which is, the computation of the rotation angles required to build a given polynomial is much more efficient and conceptually simpler within our framework.

\section{Calculating Rotation Angles}\label{sec:CalculatingPhases}
In this section, we first provide a simple method to calculate the parameters of our algorithm given desired polynomials $P$ and $Q$. This is a significant advance because existing approaches for computing the rotation angles in QSP are quite involved. This approach on the other hand is elementary and as a result, allows the approach in this paper to be more easily understood by students and more easily deployed in practice. We then provide an extremely efficient optimization algorithm for finding a corresponding $Q$ given a desired $P$ since in practice one is often only interested in implementing a particular polynomial $P$. We also provide astonishing numerical results that showcase the efficiency of our optimization method. Notably, we're able to find $Q$ for randomly chosen polynomials $P$ of up to $2^{24}\sim 16.8$ million degrees in less than $40$ seconds on an A100 GPU. This is a remarkable advancement given that previous methods have only been able to find rotation angles for polynomials of up $\sim 10^4$ degrees \cite{dong2021efficient}.\\

Suppose $P$ and $Q$ are given to us as a $2\times d$ matrix $S$ where the first row contains the coefficients of $P$ and the second row contains coefficients of $Q$. Now we'll use the induction step in the proof of ~\ref{thm:PolynomialOfUnitaries} to formulate a simple recursive algorithm by setting $\theta_d = \tan^{-1}(\frac{|b_d|}{|a_d|})$ and $\phi_d = \textrm{Arg}(\frac{a_d}{b_d})$, calculating $\hat{P}$ and $\hat{Q}$, and calling the function recursively. Calculating $\hat{P}$ and $\hat{Q}$ in this matrix representation corresponds to multiplying $S$ by $R(\theta_d, \phi_d, 0)$, and shifting the top row of the matrix which would correspond to multiplying by $A^{\dag}$ in our original formulation. The pseudo-code for this algorithm is shown in ~\ref{alg:pseudo}. As we can see, this algorithm is very simple and easy to understand, however, the algorithm assumes $P$ and $Q$ satisfy the requirements of theorem ~\ref{thm:PolynomialOfUnitaries}, which is not a trivial task.

\begin{algorithm}[t]
\caption{Pseudocode for the computation of the phase factors in GQSP.\label{alg:pseudo}}
\begin{algorithmic}
\State $a_d, b_d = S[0][d], S[1][d]$
\State $\theta_d = \tan^{-1}\left(\frac{|b_d|}{|a_d|}\right)$
\State $\phi_d = \textrm{Arg}\left(\frac{a_d}{b_d}\right)$
\State
\If{$d = 0$}
    \State $\lambda = \textrm{Arg}(b_d)$
    \State \Return $(\theta_0, \phi_0, \lambda)$
\EndIf
\State
\State $S = R(\theta_d, \phi_d, 0) \cdot S$
\State $\hat{S} = [S[0][1:d], S[1][0:d-1]]$
\State $\Vec{\theta}_{d-1}, \Vec{\phi}_{d-1}, \lambda = \textrm{ComputeParameters}(\hat{S}, d - 1)$
\State
\State \Return $(\Vec{\theta}_{d-1}, \theta_d), (\Vec{\phi}_{d-1}, \phi_d), \lambda$
\end{algorithmic}
\end{algorithm}

In theorem ~\ref{thm:QExistence} we showed that as long as $\forall x\in \mathbb{T}, \>|P(x)|^2 \leq 1$, there exists a $Q$ such that together with $P$ they satisfy the conditions of ~\ref{thm:PolynomialOfUnitaries}. One way to find such a $Q$ is through root finding by following the idea laid out in the proof of ~\ref{thm:QExistence}. However, root finding algorithms can be computationally expensive for polynomials of large degree, so here we lay out a simple and cheap optimization method for finding $Q$. Let $P(e^{it}) = \sum_{n = 0}^d a_n e^{int}$, and $Q(e^{it}) = \sum_{n = 0}^d b_n e^{int}$. In other words, $P$ and $Q$ are the discrete-time Fourier transformations of the coefficients $\{a_n\}_n$ and $\{b_n\}_n$. To be more rigorous $P$ and $Q$ are the Fourier transforms of impulse trains:
\begin{equation}
    P(e^{it}) = \mathcal{F}\left \{\sum_{n=-\infty}^{\infty} a_n \cdot \delta(t-\frac{n}{2\pi})\right \}
\end{equation}
\begin{equation}
    Q(e^{it}) = \mathcal{F}\left \{\sum_{n=-\infty}^{\infty} b_n \cdot \delta(t-\frac{n}{2\pi})\right \}
\end{equation}
Where $a_n = b_n = 0$ for $n > d$ or $n < 0$. The conditions of ~\ref{thm:PolynomialOfUnitaries} also require:
\begin{equation}
    \|P(e^{it})\|^2 + \|Q(e^{it})\|^2 = P(e^{it})P^*(e^{it}) + Q(e^{it})Q^*(e^{it}) = 1
\end{equation}
Then taking the inverse Fourier transform of both sides of the above equation, and applying the convolution theorem we get:
\begin{equation}
    \left (\sum_{n=-\infty}^{\infty} a_n \cdot \delta(t-\frac{n}{2\pi})\right) \star \left (\sum_{n=-\infty}^{\infty} a_n^* \cdot \delta(t+\frac{n}{2\pi})\right ) + \left (\sum_{n=-\infty}^{\infty} b_n \cdot \delta(t-\frac{n}{2\pi})\right) \star \left (\sum_{n=-\infty}^{\infty} b_n^* \cdot \delta(t+\frac{n}{2\pi})\right ) = \delta
\end{equation}
Where $\delta$ is the unit impulse, and $\star$ is the convolution operator. The above implies that if we're given $P$ and $Q$ as arrays of coefficients the following must hold:
\begin{equation}
    [a_0, a_1, ..., a_d] \star [a_d^*, a_{d-1}^*, ..., a_0^*]  + [b_0, b_1, ..., b_d] \star [b_d^*, b_{d-1}^*, ..., b_0^*] = [0, 0, ..., 0, 1, 0, ..., 0, 0]
\end{equation}
Where $[0, 0, ..., 0, 1, 0, ..., 0]$ is an array of length $2d + 1$ with 1 at the center and $d$ zeros on each side. Hence by letting $\Vec{a} = [a_0, a_1, ..., a_d]$, $\Vec{b} = [b_0, b_1, ..., b_d]$, and $\Vec{\delta} = [0, 0, ..., 0, 1, 0, ..., 0, 0]$ we have:
\begin{equation}
     \Vec{a} \star \text{reverse}(\Vec{a})^*+ \Vec{b} \star \text{reverse}(\Vec{b})^*  = \Vec{\delta}
\end{equation}
Here $\text{reverse}(\cdot)$ refers to the reversal of an array wherein the last element is first and vice versa.
Therefore, we can set up our optimization problem as:

\begin{equation}
    \text{argmin}_{\Vec{b}}\> \| \Vec{a} \star \text{reverse}(\Vec{a})^*+ \Vec{b} \star \text{reverse}(\Vec{b})^*  - \Vec{\delta}\|^2
\end{equation}
The above objective function can be evaluated in time $\mathcal{O}(d\log d)$ using FFT-based convolution algorithms. This almost linear runtime of the objective function makes the optimization extremely efficient as demonstrated in~\ref{fig:time_comparison}. Notably, it lets us find the coefficients of $Q$ for random polynomials of up to degree $2^{24}\sim 1.68 \times 10^7$ in less than $40$ seconds on an A100 GPU. This is a significant improvement over the degree 10,000 polynomials that can be achieved using existing QSP approaches via the  state-of-the-art techniques of \cite{dong2021efficient}. For completeness, we have provided our code\footnote{\href{https://github.com/Danimhn/GQSP-Code}{\faGithub} Our code for finding the complementary polynomial $Q$.} for the above optimization algorithm used to generate~\ref{fig:time_comparison}.

\begin{figure}[tp]
\centering
\includegraphics[width=0.8\textwidth]{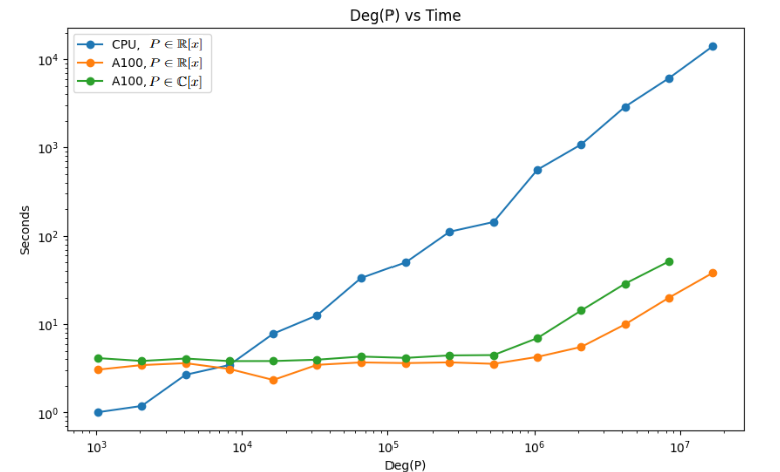}
\caption{Computational time required to complete the optimization as a function of the degree of the polynomial on CPU vs GPU. We can see that even on a CPU, we have an almost linear scaling of computational resources due to the $\mathcal{O}(n\log n)$ scaling of FFT-based convolution. Furthermore, the plot showcases the superior scalability of GPUs for the optimization problem.}
\label{fig:time_comparison}
\end{figure}

\section{\label{sec:PhaseFunctions} Phase Functions:}
A major application of quantum signal processing involves constructing phase functions of input unitaries.  Specifically, these approaches give a way to implement a transformation of an input phase via:

\begin{equation}
    U = e^{iH} \Longrightarrow e^{if(H)}
\end{equation}
This sort of transformation is the root of the exponential improvements in accuracy found for the quantum linear systems algorithm and also is the central idea behind fixed point amplitude amplification as well as recent QSP based approaches to error mitigation.
The following theorem will provide us with the building blocks to achieve such transformations in the phase.

This idea is widely used within simulation and other areas~\cite{low2017hamiltonianQSP,Yoder2014FixedpointSearch,QSVT}; however our approach provides a substantial simplification in cases where the fourier series used (equivalently a Laurent polynomial in the phase~\cite{Martyn2021GrandUnification}) is of mixed parity.  We provide a host of examples of this reasoning below and show how our generalization of phase estimation can be used to simplify the solution to the following problems.

\subsection{Application to Simulation}
The first and most obvious application of our framework is Hamiltonian simulation with the same query complexity as that of qubitization~\cite{low2019hamiltonian}.
As a first step, we need to note that the function that we wish to implement is the exponential of a cosine.  The complexity of this is given below

\begin{theorem}\label{thm:TrigPhaseFunctions}
    Given a unitary matrix $U = e^{iH}$, we can implement an $\epsilon$-approximation of $e^{it\sin H}$ and $e^{it\cos H}$ for $t\in \mathbb{R}$ using $\mathcal{O}(t + \log({\large\frac{1}{\epsilon}})/\log\log(\frac{1}{\epsilon}))$  controlled-$U$ and $2$-qubit operations while using a single ancillary qubit. 
\end{theorem}

\begin{proof}
This can be easily implemented through the  use of the Jacobi-Anger expansions:
\begin{equation}
    e^{it\cos{\theta}} = \sum_{n=-\infty}^\infty i^n J_n(t)e^{in\theta}, \>\>\>\> e^{it\sin{\theta}} = \sum_{n=-\infty}^\infty J_n(t)e^{in\theta}
\end{equation}
Where $J_n(x)$ is the $n^{th}$ Bessel function of the first kind. The coefficients of the infinite sum are known to decrease exponentially fast, specifically
\begin{equation}
    J_n(t) \in \Theta\left(\frac{\left(\frac{t}{2} \right)^n}{n!} \right) \subseteq \mathcal{O}\left(\frac{\left(\frac{et}{2} \right)^n}{n^{n+1/2}}  \right) \subseteq \mathcal{O}\left(\left(\frac{et}{2n} \right)^n  \right)
\end{equation}
Thus for any fixed $t\ge 0$ we can choose $n' \ge t$ in which case the function is exponentially decreasing for all $n\ge n'$.  Then given that we choose $n'$ to lead to exponential decay we can find for any $\epsilon>0$ that there exists $n'\in\mathcal{O}(t+\log(1/\epsilon)/\log\log(1/\epsilon))$ such that for all $n\ge n'$ $J_n(t) \le \epsilon$.
Allowing us to build $e^{it\sin H}$ and $e^{it\cos H}$ with accuracy $\epsilon$ with a polynomial of only $\mathcal{O}(t + \log({\large\frac{1}{\epsilon}})/\log\log({\large\frac{1}{\epsilon}}))$ degree. 
\end{proof}

\begin{corollary}\label{cor:sim}
Let $H=\sum_{j} \alpha_j U_j$ be a Hermitian matrix where $\alpha_j \ge 0$ and $U_j \in U(2^n)$.  Further, let $\mathtt{PREPARE} : \ket{0} \mapsto \sum_j \sqrt{\alpha_j} \ket{j} / \alpha$ and $\mathtt{SELECT} \ket{j} \ket{\psi} = \ket{j} U_j \ket{\psi}$ be unitary matrices in $U(2^m)$ and $U(2^{n+m})$ and finally let $W = - (1- 2\mathtt{PREPARE} \ket{0}\!\bra{0} \mathtt{PREPARE}^\dagger)\mathtt{SELECT}$.  Under these assumptions we can, for any $t> 0$ and $\epsilon>0$, construct a unitary matrix $V \in U(2^{n+m+1})$ that block-encodes $e^{-iHt}$ within error $\epsilon$ using $\mathcal{O}( \alpha t + \log(1/\epsilon)/\log\log(1/\epsilon))$ applications of $W$.
\end{corollary}
\begin{proof}
    Existing work~\cite{low2017hamiltonianQSP,QSVT} has shown that for each eigenvector $\ket{\lambda_j}$ of $H$ such that $H \ket{\lambda_j} = \lambda_j \ket{\lambda_j}$ there exists two eigenvectors of $W$ that can be conveniently expressed within the span of $\mathtt{PREPARE}\ket{0} \ket{\lambda_j}$ and $W\mathtt{PREPARE}\ket{0} \ket{\lambda_j}$ such that we define the orthogonal component to $\mathtt{PREPARE}\ket{0} \ket{\lambda_j}$ to be $\ket{\phi_j}$ and also define the action of $W$ restricted to this two dimensional subspace is $W_{\lambda_j}$ (taking the former vector to correspond to $\ket{0}$ and the latter to $\ket{1}$) 
    \begin{equation}
    W_{\lambda_j} = \begin{bmatrix} \frac{\lambda_j}{\alpha} & \sqrt{1-\lambda_j^2/\alpha^2} \\ -\sqrt{1-\lambda_j^2/\alpha^2} & \frac{\lambda_j}{\alpha} \end{bmatrix}
    \end{equation}
    Next note that for this unitary we have that the eigenvalues of this operation are $e^{\pm i \arccos(\lambda_j /\alpha)}$ and let us denote the eigenvectors of this to be $\ket{\phi_j^{\pm}}$.  Then setting $A = \ket{0}\!\bra{0}\otimes W + \ket{1}\!\bra{1} \otimes I^{\otimes n}$ we can apply Corollary~\ref{cor:PossiblePs} to perform for any degree $d$ polynomial $P$ such that $\max \|P(U)\| \le 1$
    \begin{equation}
        W_{\lambda_j} \mapsto \begin{bmatrix} P(W_{\lambda_j}) & \cdot \\ \cdot & \cdot \end{bmatrix}.
    \end{equation}
    In our case we wish to choose $P(U) = e^{-i\alpha t H}$. We can then use the result of Theorem~\ref{thm:TrigPhaseFunctions} to find an expansion within error $\epsilon$ using $\mathcal{O}( \alpha t + \log(1/\epsilon)/\log\log(1/\epsilon))$ queries to the oracle $W$.  We can then see that by making this transformation we can map
    \begin{equation}
        W_{\lambda_j} \mapsto \begin{bmatrix} e^{-i \alpha t \cos( \arccos(\lambda_j/\alpha))} &0 \\ 0& e^{-i \alpha t \cos(-\arccos(\lambda_j/\alpha))} \end{bmatrix} =\begin{bmatrix} e^{-i \lambda_j t} &0 \\ 0& e^{-i \lambda_j t} \end{bmatrix} := F_{\lambda_j}.
    \end{equation}
    Now let us consider the input vector $\mathtt{PREPARE}\ket{0} \ket{\lambda_j} = a \ket{\phi_j^+} + b \ket{\phi_j^-}$.  Then we have that the transformed block encoding obeys
    \begin{equation}
        (\bra{0}\otimes I)F_{\lambda_j}(\mathtt{PREPARE}\ket{0} \ket{\lambda_j}) = a e^{-i \lambda_j t}\ket{\phi_j^+} + b e^{-i \lambda_j t} \ket{\phi_j^-} = e^{-i \lambda_j t} \mathtt{PREPARE} \ket{0} \ket{\lambda_j},
    \end{equation}
    which is logically equivalent under local isometries to the evolution of the system.

    As this operation works for any two-dimensional matrix $W_{\lambda_j}$ it is straightforward to note that the same procedure must also work by linearity on $\bigoplus_j W_{\lambda_j}$ as well. Specifically, let $F = \oplus_j F_{\lambda_j}$ be the operation that we get when applying the above transformation to the entire Hilbert space of $H$ and the ancillary qubit.  Let $\ket{\psi}=\sum_j a_j \ket{\lambda_j}$ then
    \begin{equation}
        F \,\mathtt{PREPARE} \ket{0} \ket{\psi} = \sum_{j} F_{\lambda_j}\mathtt{PREPARE}\ket{0} \ket{\lambda_j} = \mathtt{PREPARE}\ket{0}e^{-iHt}\ket{\psi}.
    \end{equation}
    In order to estimate the cost of this procedure, let us consider the number of queries made to $\mathtt{PREPARE}$ and $\mathtt{SELECT}$.  An invocation of the $W$ requires three queries.  From Theorem~\ref{thm:TrigPhaseFunctions} we can perform the singular value transformation with a number of applications of $W$ that is in $\mathcal{O}( \alpha t + \log(1/\epsilon)/\log\log(1/\epsilon))$.  The claimed result then follows by multiplication of this cost by the $\mathcal{O}(1)$ cost of implementing $W$.
\end{proof}

 {Next we use the fact that this polynomial can be constructed using Theorem~\ref{thm:PolynomialOfUnitaries} we see that a degree $\mathcal{O}(\log(\frac{1}{\epsilon}))$ polynomial can be implemented using $\mathcal{O}(\log(1/\epsilon))$ controlled unitary operations and as for our circuits there are $\mathcal{O}(1)$ two-qubit gate applied in the circuit of Theorem~\ref{thm:PolynomialOfUnitaries} per controlled unitary and thus the result holds.
 The above algorithm is known to be optimal as improvements upon this scaling would lead to an algorithm that could compute the parity function using fewer than $\mathcal{O}(N)$ queries to the underlying quantum bits~\cite{berry2015hamiltonian}. Importantly, however, our approach can be used to simplify the logic of Hamiltonian simulation by removing the need to combine polynomials of different parity.

Next, we will use $e^{it\sin H}$ and $e^{it\cos H}$ as our building blocks to build other phase functions. To this aim, we invoke one of the main technical results of [\cite{van2020quantum}, Lemma 37] about approximating smooth functions by Fourier series.

\begin{theorem}\label{thm:LowWeightFourierSeries}

\noindent  Let $\delta, \epsilon \in (0, 1)$ and $f : \mathbb{R} \mapsto \mathbb{C}$ s.t. $|f(x) - \sum_{k=0}^K a_k x^k|\leq \textrm{\Large$\frac{\epsilon}{4}$}$ for all $x\in [-1+\delta, 1-\delta]$. Then $\exists c \in \mathbb{C}^{2M+1}$ s.t:

\begin{equation}\label{eq:LowWeightFourierSeries}
    |f(x) - \sum_{m=-M}^M c_m e^{\frac{i\pi m}{2}x}|\leq \epsilon
\end{equation}
 
\normalsize
for all $x\in [-1+\delta, 1-\delta]$, where $M = \max(2\lceil \log(\frac{4\|a\|_1}{\epsilon})\frac{1}{\delta}\rceil, 0)$
\end{theorem}A notable property of the prior result is that the bounds on the Fourier
series do not depend on the degrees of the polynomials terms. This can however be expected
since the terms that have large degree make negligible contributions due to the restricted domain $x \in [-1 + \delta, 1 - \delta]$, and therefore we can drop them without loss of accuracy.
{Note that these results have been traditionally used in conjunction with LCU methods to implement arbitrary functions of a generator~\cite{QSVT}.  However, here the main difference between the LCU results and ours is that our approach does not require scaling by the inverse 1-norm of the coefficients $\Vec{c}$ or the additional poly-logarithmic memory. The 1-norm improvement of our algorithm is because LCU is a more general algorithm that builds a linear combination of arbitrary unitaries and not just the powers of the same unitary, because of this LCU puts more restrictions on the probability of success.}
{As a straightforward example of this approach, let us consider the case of fractional queries to a unitary matrix.}


\subsection{Application to Fractional Queries}
{Fractional queries are a generalization of the following problem listed by Scott Aaronson as one of “The ten most annoying questions in quantum computing”: given a unitary $U$, can we implement $\sqrt{U}= e^{i\frac{H}{2}}$?~\cite{Scott_aaronson} Or more generally $U^t = e^{itH}$ for $t\in (0,1)$? Sheridan et al. \cite{Sheridan_2009} first gave an algorithm to implement the fractional power of a unitary matrix that runs in $\mathcal{O}(max\left({\large\frac{1}{\delta}}, {\large\frac{1}{\epsilon}}\right)\log({\large\frac{1}{\epsilon})})$ which was then improved upon exponentially in terms of the error dependence by \cite{QSVT} to run in $\mathcal{O}({\large\frac{1}{\delta}}\log({\large\frac{1}{\epsilon})})$.  Sheridan et al. \cite{Sheridan_2009} also proved a lower bound on this problem, which shows that the $\delta$ dependence of this algorithm is actually optimal. We provide an improvement of this algorithm below which leads to an additive logarithmic factor in $1/\epsilon$ rather than a multiplicative logarithmic factor. Hence, our algorithm here not only improves upon previous techniques, but it is in fact provably optimal up to double logarithmic factors.}

\begin{theorem}\label{thm:FractionalQueries}
\noindent  Let $U = e^{iH}$ be a unitary operator, let $t,\epsilon \in (0,1)$, and assume $\sigma_{\min}(H) \geq \delta$ where $\sigma_{\min}(H)$ is the minimum singular value of $H$. We can then implement an $\epsilon$-approximation of $U^t = e^{itH}$ with $\mathcal{O}\left(\frac{1}{\delta} + \log(\large\frac{1}{\epsilon})\right)$ uses of controlled-$U$ and $2$-qubit gates using $\mathcal{O}(\log(1/\delta))$ ancilla qubits.  Further, no algorithm is possible that solves the algorithm using $o(\frac{1}{\delta})$. 
\end{theorem}
\begin{proof}
First notice that for $x\in [0, 2\pi]$:
\begin{equation}
    e^{itx} = 
    \begin{cases}
        e^{it\arccos\left(\cos\left(x\right)\right)}, & 0\leq x\leq\pi\\
        e^{it(2\pi-\arccos\left(\cos\left(x\right)\right))}, & \pi < x < 2\pi\\
    \end{cases}
\end{equation}
We hope to use ~\ref{thm:LowWeightFourierSeries} to build $e^{it\arccos(y)}$, with $y = \cos(x)$ and make the appropriate phase shifts in the second half of the domain to build $e^{itx}$. However, notice that $y = \cos(x) \in [-1, 1]$ which blows up the number of terms in ~\ref{thm:LowWeightFourierSeries}. So the idea here is to write a more fine-grained piecewise representation of $e^{itx}$ such that $y\in [-1/\sqrt{2}, 1/\sqrt{2}]$, hence getting rid of the $\delta$ dependence in ~\ref{thm:LowWeightFourierSeries}, so the approximation part of the algorithm runs in $\mathcal{O}(\log(1/\epsilon))$. We would still need to distinguish the states at the branch-cut which depends on $\delta$, however, this is a separate part of the algorithm and that is why we end up with a $\mathcal{O}\left(\frac{1}{\delta} + \log(\large\frac{1}{\epsilon})\right)$ runtime instead of the previous best algorithm which achieves a $\mathcal{O}\left(\frac{1}{\delta}\log(\large\frac{1}{\epsilon})\right)$. We can write this more fine-grained piecewise representation of $e^{itx}$ as:\\

\begin{equation}\label{eq:PiecewiseFractional}
e^{itx} = 
    \begin{cases}
        e^{it\arcsin\left(\sin\left(x\right)\right)}, & 0\le x\le\frac{\pi}{4}\\
        e^{it\arccos(\cos(x))}, & \frac{\pi}{4}\le x\le\frac{3\pi}{4}\\
        e^{it(\pi-\arcsin\left(\sin\left(x\right)\right))}, & \frac{3\pi}{4}\le x\le\frac{5\pi}{4}\\
        e^{it(2\pi -\arccos(\cos(x)))}, &\frac{5\pi}{4}\le x\le\frac{7\pi}{4}\\
        e^{it(2\pi+\arcsin\left(\sin\left(x\right)\right))}, &\frac{7\pi}{4}\le x < 2\pi
    \end{cases}
\end{equation}

In order to implement an $\epsilon$-approximation of $U^t = e^{itH}$, we first perform a $\delta$-precise phase estimation, and use the first 3 qubits storing the phase estimation results to figure out which $8^{th}$ of the unit circle we are in, and apply the corresponding transformation based on ~\eqref{eq:PiecewiseFractional}. Since by assumption the smallest eigenvalue of $H$ in terms of norm is $\geq \delta$, no eigenvector will be miss categorized at the branch cut. Furthermore, it doesn't matter if we miss categorize the eigenvalues at the other boundary points since the piecewise function in ~\eqref{eq:PiecewiseFractional} is continuous at all other boundaries, and thus applying the neighboring piece will still result in a correct transformation. Then since a $\delta$-precise phase estimation has query complexity $\mathcal{O}\left({\large\frac{1}{\delta}}\right)$, it remains to perform an $\epsilon$-approximations of $e^{it\arcsin(\sin(x))}$ and $e^{it\arccos(\cos(x))}$ using $\mathcal{O}\left( \log({\large\frac{1}{\epsilon})}\right)$ controlled-$U$ and $2$-qubit gates.\\

First let us look at the Taylor series of $e^{it \arcsin(x)}$. One
can see that the 1-norm of the coefficients of the Taylor series of $t \arcsin(x)$ is $|t| \arcsin(1) = |t|\frac{\pi}{2}$. Therefore, for $t\in[\frac{-2}{\pi}, \frac{2}{\pi}]$
 we get that the 1-norm of the Taylor series of $e^{it \arcsin(x)}$ is $\leq e^1 = e$. Also notice that $\|\sin(x)\|$ and $\|\cos(x)\|$ are both $\leq \frac{1}{\sqrt{2}}$ on the domains where they're being used in ~\eqref{eq:PiecewiseFractional}, hence eliminating the dependence on $\delta$ for theorem \ref{thm:LowWeightFourierSeries}. Therefore, for $t\in[\frac{-2}{\pi}, \frac{2}{\pi}]$, we can write:

 \begin{equation}
    \left|e^{it\arcsin\left(\sin\left(x\right)\right)} - \sum_{m=-M}^M c_m e^{\frac{i\pi m}{2}\sin\left(x\right)}\right|\leq \epsilon
\end{equation}
For $M\in\mathcal{O}\left( \log({\large\frac{1}{\epsilon})}\right)$. Hence it suffices to $\frac{\epsilon}{M}$-approximate each term $e^{\frac{i\pi m}{2}\sin\left(x\right)}$, which by ~\ref{thm:TrigPhaseFunctions} requires a polynomial of degree at most $\mathcal{O}(M + \log({\large\frac{M}{\epsilon}})) \sim \mathcal{O}(\log({\large\frac{1}{\epsilon}}) + \log({\large\frac{\log({1/\epsilon)}}{\epsilon}})) \subseteq \mathcal{O}(\log({\large\frac{1}{\epsilon}}))$. Notice that the sum of all the polynomials can be written as a single polynomial, hence at the end we only require a single polynomial of degree $\mathcal{O}(\log({\large\frac{1}{\epsilon}}))$. Furthermore, $\arccos(x) = \frac{\pi}{2} - \arcsin(x)$, thus building $e^{it\arccos(\cos(x))}$ will have the same runtime as $e^{it\arcsin(\sin(x))}$. Then since For $t\in [-1,1]$, $U^t$ can then be implemented by applying $U^{\frac{t}{2}}$ twice, $U^t = e^{itH}$ can be implemented with complexity $\mathcal{O}\left(\frac{1}{\delta} + \log(\large\frac{1}{\epsilon})\right)$ for all $t\in[-1,1]$.

To show optimality, it is known from \cite{Sheridan_2009} that $o(1/\delta)$ scaling is impossible for fractional queries.  This shows that the only remaining question is whether the $\epsilon$ scaling is optimal.  
\end{proof}

\subsection{Application to Square-Roots / Bosonic Simulation}
Next let us consider building a series expansion for the square root Phase Function of a unitary operator.  Without loss of generality, every unitary matrix can be expressed as $e^{iH}$ for Hermitian $H$ and our aim is to examine as an example a method for constructing $e^{it\sqrt{H}}$.  The reasons why we would want to do this are many, but the simplest example involves implementing the exponential of a bosonic operator of the form $e^{i (a^\dagger + a) t}$ where $a = \sum_{j\ge 0} \sqrt{j}\ket{j+1}\!\bra{j}$.  Conventionally this operator is implemented using an arithmetic circuit~\cite{brown2010using,somma2015quantum,shaw2020quantum}, which although polylogarithmic in depth~\cite{soeken2017hierarchical}, involves substantial constant factors and requires many ancillary qubits to carry out the reversible circuitry needed for the logic.

\begin{theorem}\label{thm:SquareRootPhaseFunction}
\noindent Let $U = e^{iH}$ be a finite-dimensional unitary operator, let $\delta,\epsilon \in (0,1)$, $t \in \mathbb{R}$, and assume $\sigma_{\min}(H) \geq \delta$ where $\sigma_{\min}(H)$ is the minimum singular value of $H$. Then we can implement an $\epsilon$-approximation of $e^{it\sqrt{H}}$ using $O\left(\textbf{\large$\frac{1}{\delta}$}\left(\log(\textbf{\large$\frac{1}{\epsilon}$}) + \vert t\vert \right)\right)$ uses of controlled-$U$ and $2$-qubit gates using $\mathcal{O}(\log{1/\delta})$ ancilla qubits.
\end{theorem}

\begin{proof}
In order to build $e^{it\sqrt{H}}$, we will build $e^{i\sqrt{2\pi}t\sqrt{y+1}}$ using Theorem~\ref{thm:LowWeightFourierSeries} and sub-in $y = \frac{x}{2\pi} - 1$ to ensure $y\in [-1+\mathcal{O}(\delta), 1-\mathcal{O}(\delta)]$. Subbing this into ~\eqref{eq:LowWeightFourierSeries} we get:
\begin{equation}
    |e^{it\sqrt{x}} - \sum_{m=-M}^M c_m e^{\frac{i\pi m}{2}(\frac{x}{2\pi} - 1)}|\leq \epsilon
\end{equation}
This is equivalent to the following
\begin{equation}
    |e^{it\sqrt{x}} - \sum_{m=-M}^M c_m e^{\frac{ixm}{4}}e^{\frac{-i\pi m}{2}}|\leq \epsilon
\end{equation}
Absorbing the constant into the coefficient and re-writing the sum gives us:
\begin{equation}
    \sum_{m=-M}^M c_m' e^{\frac{ixm}{4}} = \sum_{j=0}^3 e^{\frac{ixj}{4}}\sum_{k=-M/4}^{M/4} (c_{4k+j}')e^{ikx}
\end{equation}

The result of Corollary~\ref{cor:sim} shows that an $\frac{\epsilon}{M}$-approximation of $e^{\frac{iHj}{4}}$ can be built using a polynomial of degree $\log(\textbf{\large$\frac{M}{\epsilon}$})$ of $U$ and an initial $\delta$-precise phase estimation. Thus, we can build the above approximation using a polynomial of degree 
$\mathcal{O}(\log({\large\frac{M}{\epsilon}}) + M)$ and an initial $\delta$-precise phase estimation. Since \begin{equation}M \in \mathcal{O}\left({\large\frac{1}{\delta}}\left(\log({\large\frac{1}{\epsilon}}) + \log(\|a\|_1)\right)\right),\label{eq:Mbd}\end{equation} We can implement $e^{it\sqrt{H}}$ using $\mathcal{O}(M)$ uses of controlled-$U$ and $2$-qubit gates and $\mathcal{O}(\log{1/\delta})$ ancilla qubits. So it remains to show $\log(\|a\|_1) \in \mathcal{O}(|t|)$.\\

Let $f(x) = e^{it\sqrt{x+1}}=\sum_{n=0}^\infty a_n x^n$, and $g(x) = f'(x) = \frac{it}{2\sqrt {x+1}}e^{it\sqrt{x+1}} =\sum_{n=0}^\infty b_n x^n$. So we have that $a_{n+1}=\frac{b_n}{n+1}$. Next, we'll use the Hardy inequality from~\cite{duren1970theory} which states that if $g(z) = \sum_{n=0}^\infty b_n z^n \in H^1$, then 
\begin{equation}
\sum_{n=0}^\infty \frac{|b_n|}{n+1} \le \frac{1}{2}\sup_{0<r<1}\int_0^{2\pi}|g(re^{ix})|\,dx
\end{equation}
So we get that
\begin{equation}\sum_{n = 0}^{\infty} |a_n| = \vert e^{it}\vert + \sum_{n=0}^\infty \frac{|b_n|}{n+1} \leq 1+ \frac{1}{2}\int_0^{2\pi}|g(e^{ix})| \,dx\end{equation}
Notice:
\begin{equation}\frac{1}{2}\int_0^{2\pi}|g(e^{ix})|\,dx = \frac{1}{2}\int_0^{2\pi}|\frac{it}{2\sqrt {e^{ix}+1}}e^{it\sqrt{e^{ix}+1}}|\,dx\end{equation}
Since $1+e^{i\theta}=2\cos \left(\frac{\theta}{2}\right)e^{i\theta/2}$, then $ \sqrt {1+e^{i\theta}}=\sqrt {2\cos \left(\frac{\theta}{2}\right)} e^{i\theta/4}$. We have that for $\theta \le \pi$:

\begin{equation}\Re (it \sqrt {2\cos \left(\frac{\theta}{2}\right)} e^{i\theta/4}) =-t\sin \left(\theta/4\right)\sqrt {2\cos \left(\frac{\theta}{2}\right)}\end{equation}
and similarly for $\theta \ge \pi$
\begin{equation}
    \Re (it \sqrt {2\cos \left(\frac{\theta}{2}\right)} e^{i\theta/4}) =-t\cos\left(\theta/4\right)\sqrt {2\cos \left(\frac{\theta}{2}-\pi\right)}
\end{equation}

Thus:
\begin{equation}\frac{1}{2}\int_0^{\pi}|g(e^{ix})|\,dx \leq \frac{|t|}{4} \int_0^{\pi}\frac{e^{-t\sin \left(\theta/4\right)\sqrt {2\cos \left(\frac{\theta}{2}\right)}}}{\sqrt {2\cos \frac{\theta}{2}}}d\theta \end{equation}
$e^{-t\sin \left(\theta/4\right)\sqrt {2\cos \left(\frac{\theta}{2}\right)}}$ has a maximum of $e^{\frac{\vert t\vert}{2}}$. Therefore:

\begin{equation}\frac{1}{2}\int_0^{\pi}|g(e^{ix})|\,dx \leq \frac{|t|e^{\frac{\vert t\vert}{2}}}{4} \int_0^{\pi}\frac{d\theta}{\sqrt {2\cos \frac{\theta}{2}}} \end{equation}
The integral is a constant ($\sim 3.7$), thus:
\begin{equation}\frac{1}{2}\int_0^{\pi}|g(e^{ix})|\,dx \leq \vert t\vert e^{\frac{\vert t\vert}{2}}\end{equation}
Repeating the same argument for $\theta \ge \pi$  we see that
\begin{equation}
\frac{1}{2}\int_\pi^{2\pi}|g(e^{ix})|\,dx \leq \frac{|t|}{4} \int_\pi^{2\pi}\frac{e^{-t\sin \left(\theta/4\right)\sqrt {2\cos \left(\frac{\theta}{2}-\pi\right)}}}{\sqrt {2\cos \frac{\theta}{2}}}d\theta \le |t|e^{\frac{|t|}{2}}.
\end{equation}
\begin{equation}\sum_{n = 0}^{\infty} |a_n| \leq 1 + 2\vert t\vert e^{\frac{\vert t\vert}{2}}\end{equation}
Then from the monotonicity of the logarithm function,
\begin{equation}
    \log(\sum_{n} |a_n|) \le \log(1+2|t|e^{|t|/2}) \in \mathcal{O}(|t|).\label{eq:logbd}
\end{equation}
Our result then follows by combining ~\eqref{eq:logbd} and~\eqref{eq:Mbd}
\end{proof}
In this section, we explored applications of our Generalized Quantum Signal Processing framework within the context of Phase Functions to implement unitary operators. We gave a conceptually simplified formulation of the Hamiltonian simulation algorithm within the qubitization formalism with a factor 2 reduction in the number of queries to the walk operator in cases where independent queries to $U$ and $U^\dagger$ are needed (although in many simulation examples both can be implemented using a single query~\cite{babbush2018encoding}). We further proposed an enhanced and provably optimal algorithm for implementing fractional queries of unitary operators. Lastly, we gave a novel approach for implementing square root Phase Functions of unitary matrices with applications in the implementation of bosonic operators. In the next section, we apply our method to implementing non-unitary operations by extending the concept of Fourier Decomposition to normal matrices.

\section{\label{sec:FourierBasisForMatrices} Normal Matrix Factories}
In this section, we extend the concept of Fourier decomposition to normal matrices by utilizing polynomials of a special unitary matrix. Leveraging this insight, we demonstrate that any normal matrix can be written as a polynomial of the root of unity unitary in its eigenbasis. This result will then be used as a basis for synthesizing normal matrices in various contexts. We will begin by discussing the synthesis of diagonal matrices, demonstrating that this can be accomplished with a relatively low number of single and two-qubit gates. We then further apply this method to create convolution matrices and discuss the application of our approach to implementing convolution operators and solving systems of equations involving convolution.

Given $N\in \mathbb{N}$, we define $\omega_N = e^{\frac{i2\pi}{N}}$ to be a root of unity. This mathematical construct provides a key building block in defining our main subject, the root of unity matrix $U_{\omega_\lambda}$, in relation to a given basis set. Let us consider an orthonormal basis $\lbrace \vert \lambda_j\rangle \rbrace_{j=0}^{N-1}$. With respect to this basis, we define the root of unity matrix as:

\begin{equation}\label{root_of_unity_matrix}
U_{\omega_\lambda} = \sum_{j=0}^{N-1}\omega_N^j\vert \lambda_j\rangle\langle \lambda_j\vert.
\end{equation}

Note that $U_{\omega_\lambda}$ is a unitary and diagonal matrix in the basis $\lbrace \vert \lambda_j\rangle \rbrace_{j=0}^{N-1}$. In an analogy with the complex exponential function $e^{i\theta}$, we will investigate the set of operators $\lbrace U_{\omega_\lambda}^n \rbrace_{n=0}^{N-1}$. Our goal is to demonstrate that this set constitutes a complete orthonormal operator basis for the space of all normal operators in the basis $\lbrace \vert \lambda_j\rangle \rbrace_{j=0}^{N-1}$. This fact can easily be seen to follow from the properties of the discrete Fourier transform, but a proof is given below for completeness.
\begin{lemma}\label{lem:FourierDecompositionOfMatrices}
 Given an $N \times N$ normal matrix $\mathcal{A} = \sum_{j=0}^{N-1} \lambda_j\vert \lambda_j\rangle\langle \lambda_j\vert$ for $\lambda_j \in \mathbb{C}$, $\mathcal{A}$ can be written as:\\
\begin{equation}
\mathcal{A} = \sum_{n=0}^{N-1} c_n U_{\omega_\lambda}^n
\end{equation}
where $c_n = \langle U_{\omega_\lambda}^{n},\mathcal{A}\rangle = \frac{1}{N}{\rm Tr}(\mathcal{A} \;U_{\omega_\lambda}^{-n})$.
\end{lemma}

\begin{proof}
First, note that the inner product between two powers of $U_{\omega_\lambda}$ satisfies the orthogonality property by the properties of Fourier sums:
\begin{equation}
    \langle U_{\omega_\lambda}^n, U_{\omega}^m\rangle = (\sum_{j} \omega^{j(n-m)}_N )/N = \delta_{m,n}.
\end{equation}
This implies that:
\begin{align}
    \sum_{n} c_n U_{\omega_\lambda}^n = N^{-1}\sum_{n,j}\lambda_j \omega_{N}^{-nj} \sum_k \omega_N^{kn} |\lambda_{k}\rangle\!\langle \lambda_{k}|= \sum_{jk} \delta_{jk} \lambda_j |\lambda_k\rangle\!\langle \lambda_k |=\mathcal{A}.
\end{align}
Thus the powers of the matrix $U_{\omega_\lambda}$ form a complete operator basis for normal matrices in the $\lbrace \vert \lambda_j\rangle \rbrace_{j=0}^{N-1}$ basis.  As the basis set is orthonormal in this space, it, therefore, is a complete orthonormal operator basis as claimed.
\end{proof}
Building upon the foundational results established in Lemma~\ref{lem:FourierDecompositionOfMatrices}, we are now equipped to introduce a framework for synthesizing normal matrices. We first show this in the computational basis for implementing diagonal matrices, and then extending the result to arbitrary bases. In particular we will demonstrate how we can utilize this framework in the quantum Fourier basis to implement convolution operators. 
\subsection{Synthesizing Diagonal Matrices}
In the computational basis, the construction of the root of unity matrix $U_{\omega_\lambda}$ is straightforward and can be achieved using $O(\log(N))$ many single-qubit gates, providing an efficient method to build the operator. This method's efficiency highlights the computational advantages of our framework and further establishes the value of our approach in the context of quantum information processing. From this point forward, we will denote $U_{\omega}$ as shorthand for $U_{\omega_\lambda}$ in the computational basis. In mathematical terms, this gives us:
\begin{equation}
U_{\omega}=\sum_{j=0}^{N-1}\omega_N^j |j\rangle\langle j|
\end{equation}
Where $|j\rangle$ is the binary bitstring representation of integer $j$. This succinct notation allows us to concisely express the operator in terms of the computational basis states. Lemma~\ref{lem:FourierDecompositionOfMatrices} leads to the following theorem for the synthesis of diagonal matrices.
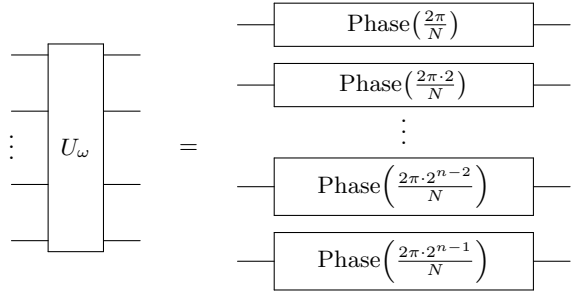
\begin{figure}[t!]
\label{fig:RootUnityComputational}
\centering
\[
\begin{array}{c}
\Qcircuit @C=1.5em @R=.7em {
    & & \multigate{7}{U_{\omega}} & \qw \\
    &  & &\\
    & & \ghost{U_{\omega}} & \qw \\
    & \vdots & \\
    &  & &\\
    & & \ghost{U_{\omega}} & \qw \\
    &  & &\\
    & & \ghost{U_{\omega}} & \qw
}
\end{array}
\quad = \quad
\begin{array}{c}
\Qcircuit @C=1.5em @R=.7em {
    & \gate{\makebox[10em]{Phase$\left(\frac{2\pi}{N}\right)$}} & \qw \\
    & \gate{\makebox[10em]{Phase$\left(\frac{2\pi \cdot 2}{N}\right)$}} & \qw \\
    & \vdots & \\
    &  & \\
    & \gate{\makebox[10em]{Phase$\left(\frac{2\pi \cdot 2^{n-2}}{N}\right)$}} & \qw \\
    & \gate{\makebox[10em]{Phase$\left(\frac{2\pi \cdot 2^{n-1}}{N}\right)$}} & \qw
}
\end{array}
\]
\caption{Quantum circuit implementing $U_{\omega}$ for a system of $n = \log N$ qubits. On the left, a multi-qubit gate $U_{\omega}$ is represented. This gate is equivalent to the application of phase gates on each qubit, as shown on the right. The phase angle doubles with each increase in the qubit's index. Here, the phase gate is defined as $\text{Phase}(\theta) = \begin{bmatrix} 1 & 0 \\ 0 & e^{i\theta} \end{bmatrix}$.}
\end{figure}

\begin{theorem}\label{thm:ArbitraryDiagonalMatrices}
    Given an $N \times N$ diagonal matrix $\mathcal{A} = P(U_\omega)$, where $\text{deg}(P) = d$, and $|P|^2\leq c$ on $\mathbb{T}$, we can implement $\frac{\mathcal{A}}{\sqrt{c}}$ using $\mathcal{O}(d\log N)$ 1 and 2-qubit gates.
\end{theorem}
\begin{proof}
    By Lemma ~\ref{lem:FourierDecompositionOfMatrices}, we know that $\mathcal{A}$ can be written as:
    \begin{equation}
        \mathcal{A} = \sum_{n=0}^{N-1} \alpha_n U_{\omega}^n
    \end{equation}
    If $|P|^2\leq c$ on $\mathbb{T}$, then $|\frac{P}{\sqrt{c}}|^2 \leq 1$ on $\mathbb{T}$. Thus, as shown by Corollary~\ref{cor:PossiblePs}, it's possible to find specific parameters $\vec{\theta}, \vec{\phi} \in \mathbb{R}^{d+1}$ and $\lambda \in \mathbb{R}$ to plug into our Generalized Quantum Signal Processing framework, such that it results in:\\
    \begin{equation} 
    \begin{bmatrix}
    \frac{\mathcal{A}}{\sqrt{c}} & .\>\>\>  \\
    .\>\>\> & .\>\>\> \\
    \end{bmatrix}
    \end{equation} 
    Therefore, we can implement $\frac{\mathcal{A}}{\sqrt{c}}$ using $\mathcal{O}(d)$ applications of controlled $U_\omega$. And as $U_\omega$ can be implemented using $O(\log(N))$ many single-qubit gates, it follows that $\frac{\mathcal{A}}{\sqrt{c}}$ can be implemented using $\mathcal{O}(d\log N)$ 1 and 2-qubit gates.
\end{proof}
This theorem paves the way for an ensuing corollary, showcasing a specific application of the above to the construction of Bit Functions. This procedure can be seen as a generalization of methods developed in~\cite{welch2014efficient}, where only implementations of unitary Bit Functions are considered.
\begin{corollary}[Bit Functions]\label{cor:BitFunctions}
    Given a function $f(x) = \sum_{n=0}^{d} \alpha_n e^{\frac{ix2\pi}{N}}$ where $|f|^2\leq 1$, we can implement the $N \times N$ diagonal matrix $\mathcal{A}$ such that:
    \begin{equation}
    \mathcal{A}|x\rangle = f(x)|x\rangle
    \end{equation}
    using $\mathcal{O}(d\log N)$ 1 and 2-qubit gates.
\end{corollary}
Upon establishing the feasibility of synthesizing diagonal matrices using a Fourier decomposition into polynomials of $U_{\omega}$, we expand our perspective beyond this initial construct. In what follows, we show that the construction of normal matrices using our method isn't confined solely to the computational basis. Indeed, given a change of basis matrix, the same principle can be effectively used to synthesize matrices in any other bases.
\begin{theorem}\label{thm:ArbitraryNormalMatrices}
    Given an $N \times N$ normal matrix $\mathcal{A} = P(U_{\omega_\lambda})=\sum_{n=0}^{d} \alpha_n U_{\omega_\lambda}^n$, diagonalized by unitary matrices $\mathcal{Q}$ and $\mathcal{Q}^{\dag}$, where implementing $\mathcal{Q}$ requires $\mathcal{O}(\chi)$ 1 and 2-qubit gates, $\text{deg}(P) = d$, and $|P|^2\leq c$ on $\mathbb{T}$, we can implement $\frac{\mathcal{A}}{\sqrt{c}}$ using $\mathcal{O}(d\log N + \chi)$ 1 and 2-qubit gates.
\end{theorem}
\begin{proof}
Utilizing Theorem~\ref{thm:ArbitraryDiagonalMatrices} we first build the polynomial in the computational basis:
\begin{equation}
    \frac{1}{\sqrt{c}}\sum_{n=0}^{d} \alpha_n U_{\omega}^n
\end{equation}
using $\mathcal{O}(d\log N)$ 1 and 2-qubit gates, and then apply the change of basis matrices to get:
\begin{equation}
    \frac{1}{\sqrt{c}}\mathcal{Q}^{\dag}(\sum_{n=0}^{d} \alpha_n U_{\omega}^n)\mathcal{Q}
\end{equation}
\begin{equation}
    = \frac{1}{\sqrt{c}}\sum_{n=0}^{d} \alpha_n \mathcal{Q}^{\dag}U_{\omega}^n\mathcal{Q}
\end{equation}
\begin{equation}
    = \frac{1}{\sqrt{c}}\sum_{n=0}^{d} \alpha_n U_{\omega_\lambda}^n=\frac{\mathcal{A}}{\sqrt{c}}
\end{equation}
Then since implementing $\mathcal{Q}$ requires $\mathcal{O}(\chi)$ 1 and 2-qubit gates, our total cost will be $\mathcal{O}(d\log N + \chi)$.
\end{proof}
\noindent We now give an example of this in the Fourier basis (i.e. $\mathcal{Q} = QFT$) to build convolution operators.
\subsection{Synthesizing Convolution Matrices}
In this subsection, we explore the synthesis of convolution matrices, which play a critical role in signal processing, image filtering, and numerous other applications.  Recent work has considered the application of linear combinations of unitaries to implement circulant matrices~\cite{zhou2017efficient} but QSP based methods have not yet been fully explored.  Our approach to synthesizing circulant matrices is based on GQSP and uses, in particular, the relationship between these matrices and circular convolutions. Specifically, a circulant matrix can be completely defined by a single vector, $\vec{c}$. This vector forms the first column of the matrix, and the remaining columns are each cyclic permutations of $\vec{c}$, each with an offset equal to the respective column index. A circulant matrix is shown below:
\begin{equation}
    \begin{bmatrix}
        c_0      & c_{N-1} & \cdots  & c_2     & c_1     \\
        c_1      & c_0     & c_{N-1} &         & c_2     \\
        \vdots   & c_1     & c_0     & \ddots  & \vdots  \\
        c_{N-2}  &         & \ddots  & \ddots  & c_{N-1} \\
        c_{N-1}  & c_{N-2} & \cdots  & c_1     & c_0     \\
    \end{bmatrix}
\end{equation}
For a circulant matrix $C$, given by the vector $\vert c\rangle = \sum_{j = 0}^{N-1} c_j \vert j\rangle$, and another vector $| \psi \rangle = \sum_{j = 0}^{N-1} x_j | j\rangle$, we can perform a convolution of $\vert \psi\rangle$ by $\vert c\rangle$ simply by multiplying $\vert \psi\rangle$ with the matrix $C$. This operation is critical for many signal processing applications:
\begin{equation}
    C\vert \psi \rangle = \sum_{j = 0}^{N-1} (\psi*c)_j \vert j\rangle
\end{equation}
The resulting $(\psi*c)_j$ of the convolution is defined as:
\begin{equation}
    (\psi*c)_j = \sum_{k = 0}^{N-1} c_k x_{[j-k \>\>mod\>\> N]}
\end{equation}
A well-known characteristic of circulant matrices is that they can be defined by an associated polynomial of the cyclic permutation matrix $\mathcal{P}$. This polynomial association is particularly beneficial when constructing circuits for quantum operations, as it allows for a straightforward definition and synthesis of the required matrix:
\begin{equation}\label{eq:CirculantPoly}
    C = c_0 I + c_1 \mathcal{P} + c_2 \mathcal{P}^2 + \dots + c_{n-1} \mathcal{P}^{n-1}
\end{equation}
The cyclic permutation matrix $\mathcal{P}$ is defined as:
\begin{equation}
\mathcal{P} = \begin{bmatrix}
 0&0&\cdots&0&1\\
 1&0&\cdots&0&0\\
 0&1&\ddots&\vdots&\vdots\\
 \vdots&\ddots&\ddots&0&0\\
 0&\cdots&0&1&0
\end{bmatrix}
\end{equation}
We can equivalently express the cyclic permutation matrix $\mathcal{P}$ as:
\begin{equation}\label{eq:ConvP}
    \mathcal{P} = \sum_{j = 0}^{N-1} \vert j+1 \>\>{\rm mod}\>\> N\rangle\langle j\vert
\end{equation}
The operator $\mathcal{P}$ is a cyclic adder which can be diagonalized using the Quantum Fourier Transform (QFT):
\begin{equation}
    \mathcal{P} = QFT^{\dag}U_{\omega}QFT
\end{equation}
Then by Theorem ~\ref{thm:ArbitraryNormalMatrices}, we get the following lemma which provides bounds on the circuit size needed to construct the circulant matrix in a $1-$ and $2-$qubit gate library.  Note that the dependence on an error target $\epsilon$ is absent here because of the assumed form of the polynomial decomposition and also because rotation synthesis is not needed in this (continuous) gate set.
\begin{lemma}\label{lem::Circulant}
\noindent Given an $N \times N$ circulant matrix $C = \sum_{n=0}^d c_n \mathcal{P}^n$, we can build $C$ (normalized) using only $\mathcal{O}(d \log{N} + \log^2N)$ 1 and 2-qubit gates.
\end{lemma}
This lemma enables us to reach a valuable conclusion:
\begin{theorem}\label{col:Convolution}

\noindent Given $\vert \psi \rangle = \sum_{j = 0}^{N-1} x_j \vert j\rangle$ and a filter $F = \{a_k\}_{k=-d}^d$, we can convolve $\psi$ with $F$:
\begin{equation}
\vert \psi * F \rangle = \sum_{m = 0}^{N-1} (\psi * F)_m \vert m\rangle, \>\>\>\>\>\> \mbox{where} \>\>\>\>
(\psi * F)_m = \sum_{k = -d}^{d} a_k x_{[m-k \>\>mod\>\> N]}
\end{equation}
using only $O(d \log{N} + \log^2N)$ 1 and 2-qubit gates.
\end{theorem}
A powerful application of the previous result is solving systems of equations that are expressed as discrete convolutions. For instance, consider a system of equations represented as
\begin{equation}
    x * c = b
\end{equation}
Such equations are significant because we often have a known filter function that is applied to an input and we would like to have an efficient method for inverting such a transformation to find the input that is partially obscured by the convolution.  To see how this can be attained, we can see from the discrete convolution theorem that the original convolutional equation can be re-expressed as a linear system via
is equivalent to:
\begin{equation}
    Cx = b
\end{equation}
where $C$ is a circulant matrix. If $C$ is invertible, we can then build $C$ using Lemma~\ref{lem::Circulant} and invert it using the quantum matrix inversion algorithm to solve the system of equations.  The cost of doing so within error $\epsilon$ is then $\tilde{\mathcal{O}}(d^2 \kappa^2 \log^4(N) \log(1/\epsilon)))$ where $\kappa$ is the condition number of the circulant matrix $C$ using the inversion method of~\cite{childs2017quantum}.  This constitutes a potential exponential speedup for computing moments of the vector $x$ over the na\"ive method of solving for $x$ using matrix inversion and computing the mean from the result.  Our approach allows us then to directly synthesize such matrices through the polynomial series definition of circulant matrices after diagonalization through the Quantum Fourier transform.

\section{Conclusion}
This paper introduces a substantial advancement to quantum signal processing -- the Generalized Quantum Signal Processing (GQSP) method. Unlike traditional QSP frameworks, our method employs a pair of rotations instead of solely relying on either $Z-$ or $X-$rotations for signal processing operations. This strategic modification enables us to move beyond the limitations of the original QSP framework.

Another essential contribution of our GQSP method is the significant simplification it offers in the computation of phase angles compared to existing methods. In instances where both $P$ and $Q$ are known, we introduce a straightforward recursive formula for the angles. This substantial pedagogical improvement addresses one of the key challenges in teaching QSP methods, as the traditional techniques for finding the polynomial function in QSP can be difficult to convey. Our approach simplifies this complex aspect of QSP, making it much more accessible. Additionally, we introduced an efficient optimization algorithm for computing phase angles when only $P$ is known, but $Q$ is not. Our tests showed that our method can compute polynomials of degree greater than $10^7$ in under a minute, an impressive computational efficiency when compared to the $10^4$ degree polynomials that can be achieved using existing QSP approaches via the state-of-the-art techniques of \cite{dong2021efficient}.

In this paper, we explored several applications of our GQSP methodology. We presented an optimized algorithm for quantum fractional queries, along with a simplified technique for performing Hamiltonian simulation using qubitization. We proposed methods for calculating phase functions, such as the square root, and unveiled new methodologies for synthesizing circulant matrices and performing convolution operations. 

As we look forward, our work reveals several potential areas for further exploration. A primary question is how to adapt the principles of our approach to multivariable QSP~\cite{rossi2022multivariable}. Also, our focus on QSP suggests the potential of applying similar concepts to quantum singular value transformation for transforming block-encoded non-square matrices. In essence, this paper represents a significant progression in the QSP/QSVT framework. The exploration and understanding of the broad range of opportunities offered by these techniques promise to be a primary focus of research in quantum algorithms in the years to come.
\acknowledgements
NW would like to acknowledge funding for this work from Google Inc. This material is based upon work supported by the U.S. Department of Energy, Office of Science, National Quantum Information Science Research Centers, Co-design Center for Quantum Advantage (C2QA) under contract number DE-SC0012704 (PNNL FWP 76274).  We would like to thank Dominic Berry for useful feedback.

\nocite{*}

\bibliographystyle{unsrt}
\bibliography{bib.bib}

\end{document}